\documentclass[conference]{IEEEtran}
\newcommand{\bl}{\textcolor{black}}

\usepackage{amsmath,amsfonts,amssymb,amsthm}
\usepackage{algorithmic}
\usepackage{array}
\usepackage{subcaption}
\usepackage{graphicx}
\usepackage{multirow}
\usepackage{textcomp}
\usepackage{stfloats}
\usepackage{url}
\usepackage{verbatim}
\usepackage{graphicx}
\usepackage{color, soul}
\usepackage{bm}
\usepackage{hyperref}

\usepackage{cite}
\usepackage{xcolor}
\usepackage{xpatch}

\newtheorem{theorem}{Theorem}

\usepackage{algorithmic}
\usepackage[linesnumbered,ruled]{algorithm2e}
\usepackage{cuted}

\def\BibTeX{{\rm B\kern-.05em{\sc i\kern-.025em b}\kern-.08em
    T\kern-.1667em\lower.7ex\hbox{E}\kern-.125emX}}
\usepackage{balance}

\begin{document}

\title{User-centric Service Provision for Edge-assisted Mobile AR: A Digital Twin-based Approach}

{
    \author{
    \IEEEauthorblockN{Conghao~Zhou\IEEEauthorrefmark{1},~Jie~Gao\IEEEauthorrefmark{2},~Yixiang Liu\IEEEauthorrefmark{3},~Shisheng~Hu\IEEEauthorrefmark{1},~Nan~Cheng\IEEEauthorrefmark{4},~Xuemin (Sherman)~Shen\IEEEauthorrefmark{1}
        \IEEEauthorblockA{\IEEEauthorrefmark{1}Department~of~Electrical~and~Computer~Engineering,~University~of~Waterloo,~Canada
        \\\IEEEauthorrefmark{2}School of Information Technology, Carleton University,~Canada
        \\\IEEEauthorrefmark{3}School of Computer Science and Technology,~Xidian University,~China     
        \\\IEEEauthorrefmark{4}School~of~Telecommunications Engineering and the State Key Laboratory of ISN,~Xidian University,~China
        \\\{c89zhou, s97hu, sshen\}@uwaterloo.ca, jie.gao6@carleton.ca, yxliu21@stu.xidian.edu.cn, dr.nan.cheng@ieee.org}
            }
}

\maketitle

\begin{abstract}

\bl{Future 6G networks} are envisioned to support mobile augmented reality (MAR) applications and provide customized immersive experiences for users via advanced service provision. In this paper, we investigate user-centric service provision for edge-assisted MAR to support the timely camera frame uploading of an MAR device by optimizing the spectrum resource reservation. To address the challenge of non-stationary data traffic due to uncertain \bl{user movement} and the complex camera frame uploading mechanism, we develop a digital twin (DT)-based data-driven approach to user-centric service provision. Specifically, we first establish a hierarchical data model with well-defined data attributes to characterize the impact of the camera frame uploading mechanism on the user-specific data traffic. We then design an easy-to-use algorithm to adapt the data attributes used in traffic modeling to the non-stationary data traffic. We also derive a closed-form service provision solution tailored to data-driven traffic modeling with the consideration of potential modeling inaccuracies. \bl{Trace-driven} simulation results demonstrate that our DT-based approach for user-centric service provision outperforms conventional approaches in terms of adaptivity and robustness.

\end{abstract}


\section{Introduction}

Augmented reality (AR), falling under the extended reality spectrum, enables integrating virtual objects seamlessly into the physical surroundings of human users~\cite{shen2023toward}. Driven by the increasing demand for immersive experiences, mobile AR (MAR) accessible on mobile or portable devices such as smart glasses are gaining widespread attention as one of the emerging applications in the 6G era. All MAR applications need the procedure of device pose tracking, which is fundamental for the effective 3D alignment of virtual objects with physical environments but resource-intensive~\cite{chen2023networked}. Solely enabling device pose tracking poses a key challenge for current MAR devices due to their resource limitations such as limited battery power. To realize practical implementation of MAR, edge-assisted MAR leveraging the resources of edge servers through wireless links becomes a promising paradigm~\cite{chen2023adaptslam}.

An advanced feature that \bl{future} 6G networks may enable for edge-assisted MAR is achieving user-centric service provision to support timely user interactions between MAR devices and edge servers. While service provision is a classic research topic from the networking perspective~\cite{sun2024knowledge}, MAR applications featuring extensive human involvement that deeply affects resource demands, thereby necessitating more effective resource management strategies in 6G networks \bl{due to the following two reasons}. First, \bl{differences in user movement such as head turning} result in significantly distinctive network resource demands for different users using the same MAR application~\cite{ran2020multi}. Traditional service provision approaches relying on service-based demand modeling, e.g., video traffic modeling, fail to distinguish service demands across MAR users~\cite{navarro2020survey}, thereby compromising the flexibility of networks in supporting personalized MAR user experiences in the 6G era. Second, to deal with the uncertainties in human movement, MAR has incorporated a complex operational mechanism, e.g., simultaneous localization and mapping (SLAM)-based device pose tracking, from an application perspective to ensure immersive user experiences~\cite{campos2021orb}, which significantly complicates the demand modeling from the networking perspective. Conventional service-based demand modeling techniques struggle to capture the impact of the operational mechanism underlying MAR applications on resource demands, thereby reducing the adaptivity of service provision in accommodating user movement variations~\cite{zhou2024digital}. Therefore, a novel and advanced service provision \bl{for MAR} is essential in the 6G era.

In this paper, we investigate a service provision problem to facilitate edge-assisted MAR device pose tracking \bl{in future 6G networks}. However, two challenges arise. First, the MAR operational mechanism is highly intricate, typically involving multiple interacting functionality modules~\cite{linowes2017augmented}. The impact of multiple factors inherent in the MAR operational mechanism significantly complicates the modeling of the uplink data traffic in MAR. Second, temporal variations in user movement may lead to non-stationary uplink data traffic. For example, the data traffic load for uploading camera frames may surge intermittently \bl{due to the need of dealing with device pose tracking losses}~\cite{ran2020multi}. Such variations \bl{compromise} the effectiveness of established data traffic models due to their insufficient adaptability to uncertain user movement. 

To address these challenges, we develop a digital twin (DT)-based approach that facilitates \emph{user-centric} and \emph{data-driven} service provision to support edge-assisted device pose tracking in MAR. Specifically, we establish an MAR user DT (M-UDT) for each individual MAR device, building on our general DT framework~\cite{shen2021holistic}. The M-UDT is established by \bl{defining} a customized data model to characterize the uplink data traffic from an individual MAR device and various M-UDT functions to continuously manage the data model according to the variations in data traffic. Based on the data provided by the M-UDT, user-centric service provision decisions can be made for each MAR device. The main contributions of this paper are as follows:
    \begin{itemize}
        \item We establish a personalized hierarchical data model, organizing data attributes carefully chosen for MAR, \bl{to capture the implicit impact of the MAR operational mechanism on the uplink data traffic of an MAR user.} 

        \item We propose two machine learning-based methods with different complexities for data traffic modeling. In addition, we design an easy-to-use mechanism \bl{for switching between the two methods} to adapt to non-stationary uplink data traffic in MAR.  

        \item We derive a closed-form resource reservation solution to a service provision problem for an individual MAR device, considering potential inaccuracies in the data-driven traffic modeling, which enhances the robustness of the DT-based service provision approach.
    \end{itemize}

\section{System Model and Problem Formulation}

\subsection{Considered Scenario}

When a user runs an MAR application with an MAR device, the position and orientation (jointly referred to as \emph{3D pose}) of the MAR device change over time due to user movement. The MAR device captures camera frames periodically with a fixed frame rate and tracks its 3D pose based on the captured camera frames, which is crucial for rendering virtual objects at correct locations within the user’s field of view~\cite{campos2021orb}.

An emerging paradigm of edge-assisted device pose tracking in MAR~\cite{chen2023networked,chen2023adaptslam} is shown in Fig.~\ref{system}, wherein an MAR device and an edge server deployed at a base station (BS) collaboratively track the device pose. Specifically, the MAR device is equipped with a lightweight tracking module \bl{for real-time pose calculation}, while the edge server is equipped with a resource-intensive mapping module \bl{for the creation of} a 3D representation of the physical environment (i.e., a 3D map), which \bl{supports the device pose calculation at the MAR device}. 

Edge-assisted device pose tracking consists of four steps~\cite{ben2022edge}: (i) the MAR device selects a subset of recently captured camera frames, termed as \emph{key frames}, and uploads these key frames to the edge server over a wireless communication link; (ii) the mapping module equipped at the edge server updates the 3D map using the uploaded key frames; (iii) the edge server sends the updated 3D map back to the MAR device; and (iv) the tracking module at the MAR device leverages the updated 3D map to locally calculate the device pose \bl{for every camera frame}. The four steps iterate in device pose tracking.

    \begin{figure}[t]
        \centering
        \includegraphics[width=0.30\textwidth]{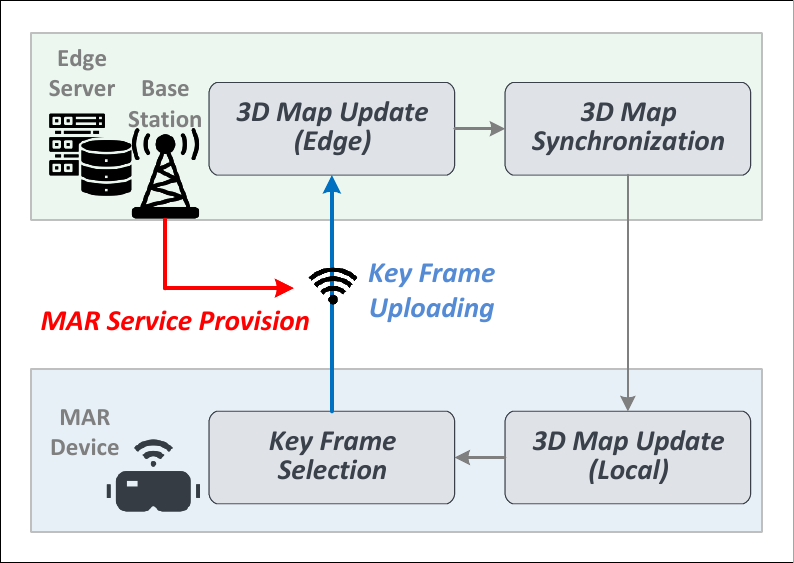}
        \caption{The considered scenario of edge-assisted MAR.}\label{system}
    \end{figure}

\subsection{Key Frame Uploading}\label{}

Let $\mathcal{F}$ denote the set of camera frames captured over the entire considered time domain. The MAR device periodically selects key frames from recently captured camera frames and uploads them to the edge server for updating the 3D map. We refer to the duration of $F$ consecutive camera frames as a time slot and denote the set of all time slots by~$\mathcal{T}$. Let~$\mathcal{F}_{t} \subseteq \mathcal{F}$ denote the set of camera frames captured during time slot~$t \in \mathcal{T}$. At the end of time slot~$t$, the MAR device determines the set of key frames for uploading, denoted by~$\mathcal{K}_{t} \subseteq \mathcal{F}_{t}$. Generally, a key frame differs sufficiently from its preceding camera frames, while there should be sufficient overlap between selected key frames~\cite{campos2021orb}. Due to uncertain user movement and/or variations in the surrounding environment, the operational mechanism of key frame selection and uploading is intricate. Considering that the number of key frames~$\tilde{k}_{t}$ may be time-varying~\cite{ben2022edge}, we model the number of key frames in each time slot as a random variable~$\tilde{k}_{t} = |\mathcal{K}_{t}|$.

Proper resource reservation for \bl{timely} key frame uploading is necessary for real-time device pose tracking. Let~$r_{t}$ denote the uplink data rate of the MAR device within time slot~$t$, given by:
    \begin{equation}\label{}
        r_{t} = b_{t} \log (1 + \gamma_{t}^\text{s}), \,\, \forall t \in \mathcal{T},
    \end{equation}
where $b_{t}$ and~$\gamma_{t}^\text{s}$ represent the amount of spectrum resource reserved to the MAR device for uplink communication and the predicted signal-to-noise ratio, respectively, in time slot~$t$. We denote the volume of data (in bits) to transmit for uploading each camera frame by~$\alpha$, assuming the same data volume for all camera frames. Given uplink data rate~$r_{t}$, the set of key frames selected for uploading in time slot~$t$ should satisfy the following constraint~\cite{atawia2016joint}:
  
    \begin{equation}\label{eq2}
        P ( T^\text{r} r_{t} \ge \alpha \tilde{k}_{t} ) \ge \varepsilon, \,\, \forall t \in \mathcal{T},
    \end{equation}
where~$T^\text{r}$ represents the maximum tolerable \bl{total transmission duration} for uploading the selected key frames \bl{before the end of each time slot}, and $\varepsilon \in [0, 1]$ represents the \bl{required reliability in MAR service provision}.

\subsection{3D Map Update \& Synchronization}\label{}

A 3D map used for edge-assisted device pose tracking consists of a set of key frames uploaded by the MAR device over time as well as the feature points (FPs), e.g., a wall corner, detected from each key frame. Given a camera frame~$f \in \mathcal{F}_{t}$, we denote the set of FPs identified in this camera frame by~$\mathcal{U}_{f}$. Since the MAR device periodically uploads newly key frames to the edge server, the 3D map maintained by the edge server changes over time. Let~$\mathcal{K}^\text{map}_{t} \subseteq \mathcal{F}$ denote the set of key frames stored in the 3D map in time slot~$t$, evolving as follows:  
    \begin{equation}\label{eq3}
        \mathcal{K}^\text{map}_{t} = \left\{ \mathcal{K}^\text{map}_{t-1} \cup \mathcal{K}_{t-1} \right\} \backslash \mathcal{C}_{t-1}, \,\, \forall t-1, t \in \mathcal{T},
    \end{equation}
where~$\mathcal{C}_{t-1} \subseteq \mathcal{K}^\text{map}_{t-1}$ represents the set of key frames removed from the 3D map~$\mathcal{K}^\text{map}_{t-1}$ maintained by \bl{the} edge server in time slot~$t-1$. 
The set~$\mathcal{K}^\text{map}_{t}$ and the set of FPs corresponding to each key frame, jointly representing the updated local 3D map, are downloaded by the MAR device. Generally, in MAR applications, selecting the set~$\mathcal{K}_{t}$ from the set of newly captured frames~$\mathcal{F}_{t}$ requires information on the updated local 3D map at time slot~$t$.

\subsection{Problem Formulation}

To efficiently support edge-assisted device pose tracking in MAR, we formulate a service provision problem with the objective of minimizing the amount of spectrum resource reserved for key frame uploading, as follows:

    \begin{subequations}\label{p1}
        \begin{align}
            \textrm{P1:} &\,\, \min_{ \{ b_{t} \}_{t \in \mathcal{T}} } \sum_{ t \in \mathcal{T} }{ b_{t} }\\
            \textrm{s.t.} &\,\, P( T^\text{r} r_{t} \ge \alpha \tilde{k}_{t} ) \ge \varepsilon, \,\, \forall t \in \mathcal{T},
        \end{align}
    \end{subequations}
where the optimization variable~$b_{t}$ corresponds to the amount of the reserved \bl{spectrum} resource for key frame uploading in each time slot. Constraint~(\ref{p1}b) ensures the transmission duration for key frame uploading. Problem~P1 is intractable since~$\tilde{k}_{t}$ is unknown~\emph{a priori}, and \bl{temporal variations in data traffic of each MAR device may be non-stationary}. \bl{Specifically,} conventional approaches fall into using either mathematical modeling or data-driven prediction, to achieve the on-demand resource reservation by accurately modeling the uplink data traffic~\cite{navarro2020survey}. However, these approaches are designed for general network resource reservation problems and, thus, may overlook the impact of the specific MAR operational mechanism~\cite{chen2023adaptslam,ben2022edge}, on uplink data traffic load. Additionally, they may struggle to adapt to non-stationary traffic variations due to using a single data traffic model.

We develop a digital twin (DT)-based approach to characterize the impact of the MAR operational mechanism on the data traffic of an individual MAR device, thereby enabling user-centric service provision.

\section{The Developed Digital Twin-based Approach}

In this section, we establish {an} MAR user DT (M-UDT) for the MAR device, and our M-UDT design evolves from the framework presented in~\cite{shen2021holistic,zhou2024digital,hu2023adaptive}. The M-UDT, comprising an \emph{MAR user profile} (MUP) and following \emph{UDT functions}, is deployed at the BS and maintained by the controller to facilitate MAR service provision. 

\subsection{Data-driven Demand Modeling Function (DMF)}

User-centric service provision requires an accurate model for capturing the uplink data traffic pattern of the individual MAR device. To obtain such \bl{a data traffic model}, we employ a Markov decision process to abstract the sequential decision making underlying the key frame uploading of the MAR device. Define state~$\mathbf{s}_{t} \in \mathcal{S}$, action~$\mathbf{a}_{t} \in \mathcal{A}$, state transition probability function~$P(\mathbf{s}_{t+1}|\mathbf{s}_{t}, \mathbf{a}_{t}) := \mathcal{S} \times \mathcal{A} \rightarrow \mathcal{S} $, and policy~$\pi(\mathbf{a}_{t}|\mathbf{s}_{t}) := \mathcal{S} \rightarrow \mathcal{A}$. We use the selected set of key frames $\mathcal{K}_{t}$ to define the action in time slot~$t$, denoted by~$\mathbf{a}_{t} = [a_{t,f}]_{\forall f \in \mathcal{F}_{t}} \in \mathcal{A}$, where $a_{t,f} = 1$ if $f \in \mathcal{K}_{t}$, and $a_{t,f} = 0$ otherwise. Given action~$\mathbf{a}_{t}$, the corresponding data traffic \bl{load} for key frame uploading can be determined. 

To model the data traffic, we denote the policy of key frame uploading \bl{that is actually used in the considered MAR application and affected by the MAR operational mechanism~\cite{zhou2024digital} by $\pi^\text{A}$.} To approximate~$\pi^\text{A}$ \bl{accurately}, states~$\mathbf{s}_{t}$ need to be carefully defined since factors influencing key frame uploading in MAR may be implicit and intricate. Therefore, we introduce two types of states for \emph{detailed} and \emph{simplified} traffic modeling, respectively. In addition to the approximation of the actual policy~$\pi^\text{A}$, the established UDT function should approximate the state transition \bl{probabilities} to support data traffic modeling over multiple time slots.

\subsubsection{Detailed Modeling}

In MAR applications, the set~$\mathcal{K}_{t}$ is determined based on the correlation among key frames in 3D map~$\mathcal{K}^\text{map}_{t-1}$ and the correlation among camera frames in set~$\mathcal{F}_{t}$. To characterize the impact of such correlations on key frame uploading, we define 3D map~$\mathcal{K}^\text{map}_{t}$ as a weighted undirected graph denoted by~$\mathcal{G}^\text{map}_{t} = (\mathcal{K}^\text{map}_{t}, \mathcal{E}^\text{map}_{t})$, where~$\mathcal{E}^\text{map}_{t}$ represents the set of edges between every pair of camera frames in~$\mathcal{K}^\text{map}_{t}$. For edge~$e = (f, f') \in \mathcal{E}^\text{map}_{t}$ connecting camera frames~$f, f' \in \mathcal{K}^\text{map}_{t}$, the weight of edge~$e$ is defined as the Jaccard coefficient~\cite{khosoussi2019reliable}:
    \begin{equation}\label{}
        \epsilon_{f, f'} = \frac{|\mathcal{U}_{f} \cap \mathcal{U}_{f^{'}}|}{|\mathcal{U}_{f} \cup \mathcal{U}_{f^{'}}|}, \,\,\,\, \forall\,\, \mathcal{U}_{f} \cup \mathcal{U}_{f^{'}} \neq \emptyset,
    \end{equation} 
where $\cap$ and $\cup$ denote the intersection and the union of two sets, respectively. The Jaccard coefficient~$\epsilon_{f, f'}$ quantifies the similarity of the two sets. If the two sets of FPs~$\mathcal{U}_{f}$ and $\mathcal{U}_{f'}$ are similar, the weight,~$\epsilon_{f, f'} \in [0,1]$ is large. Similarly, we define the graph for set~$\mathcal{F}_{t}$ as~$\mathcal{G}_{t} = (\mathcal{F}_{t}, \mathcal{E}_{t})$. We define~$\mathbf{s}^\text{d}_{t} = [\mathcal{G}^\text{map}_{t-1}, \mathcal{G}_{t}]$ as \bl{the state in the detailed modeling} and find a graph convolutional network (GCN), denoted by~$\pi^\text{d}$, with parameters~$\boldsymbol{\vartheta}^\text{d}$ to approximate policy~$\pi^\text{A}$ by minimizing the following loss function:
    \begin{equation}\label{}
        L(\boldsymbol{\vartheta}^\text{d}) = \frac{1}{|\Xi|} \sum_{(\mathbf{a}_{t},\mathbf{s}^\text{d}_{t}) \in \Xi} \left(\mathbf{a}_{t} - \pi^\text{d}(\mathbf{s}^\text{d}_{t};{\boldsymbol{\vartheta}^\text{d}}) \right)^{2},
    \end{equation}
where~$\Xi$ represents a set containing historical information on actions and states, stored in the MUP. 

\subsubsection{State Transition Modeling}

To support long-term service provision, the DMF models state transitions~$P(\mathbf{s}^\text{d}_{t+1}| \mathbf{s}^\text{d}_{t}, \mathbf{a}_{t})$. 

Due to the fact that newly arrived camera frames in~$\mathcal{F}_{t}$ do not depend on 3D map~$\mathcal{K}^\text{map}_{t}$, and~$P(\mathcal{G}^\text{map}_{t}| \mathcal{G}^\text{map}_{t-1}, \mathbf{a}_{t})$ is known according to \eqref{eq3}. Therefore, to model state transitions, we focus on approximating~$P(\mathcal{G}_{t}| \mathcal{G}_{t-1})$ by using \bl{another} GCN~$\phi(\mathcal{G}_{t-1};\boldsymbol{\theta})$ with parameters~$\boldsymbol{\theta}$. Note that this GCN needs to output only the weights of edges between camera frames, instead of raw images, which \bl{can be} categorized as the link prediction in graph theory.  

\subsubsection{Simplified Modeling}

Although the detailed modeling incorporates the impacts of 3D map and historical camera frames, \bl{excessive input data may introduce redundancy and thus decrease the modeling accuracy. For example, the procedure of key frame selection and uploading in the MAR operational mechanism for device pose tracking is simple when the variation in device pose is insignificant~\cite{campos2021orb,ben2022edge}.} To deal with this issue, we propose a simplified data-driven modeling \bl{as an alternative}. Define~$\mathbf{s}^\text{s}_{t} = [\mathbf{a}_{i}]_{\forall t- T^\text{w} \leq i < t}$ as a state in the simplified modeling at time slot~$t$, which includes the actions conducted in the \bl{preceding}~$T^\text{w}$ time slots. In this case, the approximation of the policy~$\pi^\text{A}$ can be simplified as conventional temporal sequence prediction. We build a recurrent neural network~$\pi^\text{s}$ with parameters~$\boldsymbol{\vartheta}^\text{s}$ and realize the approximation using the following loss function:
    \begin{equation}\label{}
        L(\boldsymbol{\vartheta}^\text{s}) = \frac{1}{|\Xi|} \sum_{(\mathbf{a}_{t},\mathbf{s}^\text{s}_{t}) \in \Xi} \left(\mathbf{a}_{t} - \pi^\text{s}(\mathbf{s}^\text{s}_{t};{\boldsymbol{\vartheta}^\text{s}}) \right)^{2}.
    \end{equation}
Since state~$\mathbf{s}^\text{s}_{t}$ consists of only previous actions, state transitions are straightforward and do not require additional modeling.

\subsection{Model Switching Function (MSF)}

    \begin{algorithm}[t] 
        \caption{Model Switching Method}\label{alg1}
        \LinesNumbered
        \textbf{Input:} $M$, $\delta$;\\
        \textbf{Initialization:} $h_{1} = 1, m_{1} = 0$;\\
        \For{$t \in \mathcal{T}$}
        {       
            $\Delta_{t} = |\mathcal{K}_{t-1}| - |\mathcal{K}_{t-2}|$;\\
            \eIf{$\Delta_{t} > \delta$}
                {
                    $h_{t} \leftarrow 1$; $m_{t} \leftarrow 0$;\\
                }
                {   
                    $m_{t} \leftarrow m_{t-1}+1$;\\
                    \If{$m_{t} \ge M$}
                    {
                       $h_{t} \leftarrow 0$; $m_{t} \leftarrow 0$;\\   
                    } 
                    
                }
        }
        \textbf{Output:} $h_{t}$

    \end{algorithm} 

The MSF function is designed to accurately adapt the data-driven DMF to non-stationary uplink data traffic via flexible model switching. In MAR applications, when variations in the physical environment and user movement are insignificant, the MAR operational mechanism of key frame selection and uploading is simple, leading to relatively stable uplink traffic; Conversely, a significant variation such as a variation leading to pose tracking loss generally \bl{complicates} the MAR operational mechanism, potentially resulting in bursts of key frame uploading. Define~$h_{t} \in \left\{0, 1\right\}$ as an indicator for model switching. If~$h_{t} = 1$, the detailed model is used at time slot~$t$; Otherwise, the simplified model is used. We provide an easy-to-use model switching mechanism in Algorithm~\ref{alg1} based on the temporal variation in the number of uploaded key frames. Parameters~$\delta$ and $M$ jointly determine the switching condition, which can be adjusted flexibly according to user movement and user-specific psychical environment.

\subsection{MAR User Profile (MUP)}

The MUP offers a user-centric \emph{data model} consisting of a number of data elements \bl{that are carefully defined} and organized in a structured way. The data model can implicitly characterize the complex impacts of data elements pertinent to the MAR operational mechanism on the resource demand from an individual MAR device~\cite{shen2021holistic,ma2023nomore}. The designed DMF and MSF can update the MUP via updating data elements in the data model, thereby facilitating MAR service provision.

As shown in Fig.~\ref{data_model}, we build a hierarchical data model to support MAR service provision. At the top level of this hierarchy, there is a ``user terminal'' representing an MAR device such as smart glasses. An individual MAR device consists of a number of ``functional units'', each relating to a unique functionality, e.g., tracking or rendering, in the MAR application. Each functional unit contains a set of purposefully chosen ``data attributes'' related to \bl{the MAR operational mechanism of that functional unit}. Although this paper considers service provision, for a single functional unit (i.e., device pose tracking), the data model has the flexibility and scalability to adapt to various MAR functionalities and network management objectives.

    \begin{figure}[t]
        \centering
        \includegraphics[width=0.42\textwidth]{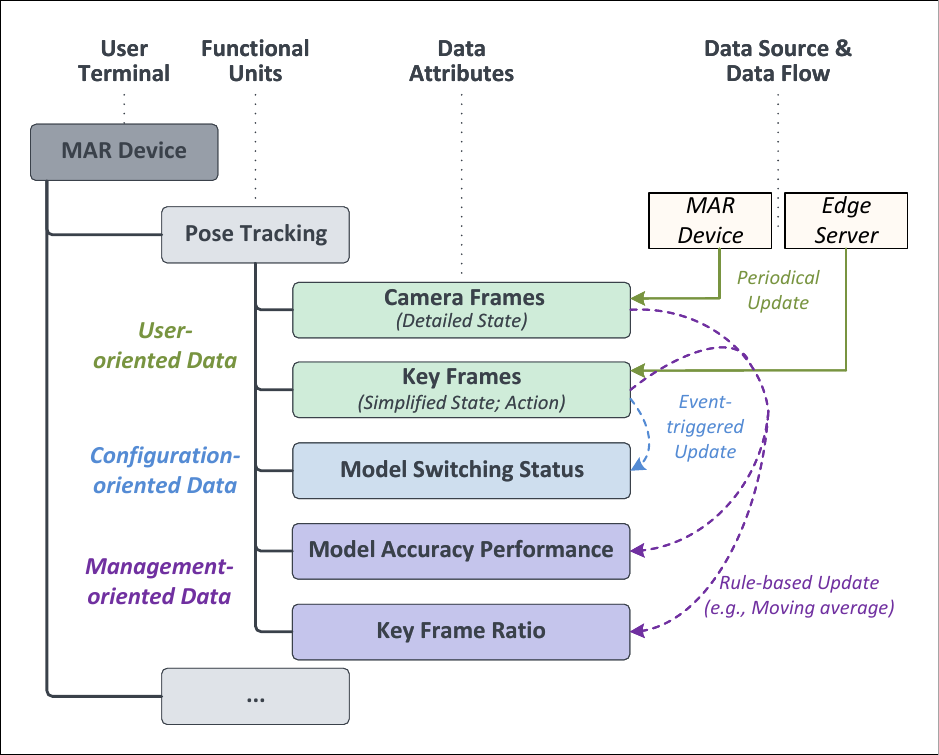}
        \caption{The hierarchical data model in the MUP.}\label{data_model}
    \end{figure}

The data flows within the UDT for MUP update vary across different data attributes depending on the purposes for which the data are used. We classify data in this MUP into three categories: i) User-oriented data, e.g.,~$\mathbf{s}^\text{d}_{t}$ and $\mathbf{a}_{t}$, that are used to characterize the service demand of an individual MAR device and can be periodically collected; ii) Configuration-oriented data, e.g.,~$h_{t}$, that are used to configure the DMF and MSF and may be updated based on the change of user-oriented data in an event-triggered way; and iii) Management-oriented data, e.g., model accuracy, that are used to enable user-centric service provision and obtained from the statistical analysis of user-oriented data given a predefined rule, which will be introduced in Subsection~\ref{sec34}.   

\subsection{M-UDT-based User-centric Service Provision}\label{sec34}

Unlike traditional mathematical models that offer a stochastic representation of data traffic to guide service provision, the M-UDT employs data-driven traffic modeling that outputs predicted data traffic volumes. Currently, neither mathematical models nor data-driven models achieve the absolute modeling accuracy~\cite{navarro2020survey}. To address the potential inaccuracies of the M-UDT in data traffic modeling, we propose a robust service provision method tailored to data-driven traffic modeling.

Define~$\hat{a}_{t,f}$ as the prediction value of~$a_{t,f}$ via the M-UDT. The optimal M-UDT-based service provision solution to Problem~P1 is as follows:
    \begin{equation}\label{}
      b^{*}_{t} = \frac{\alpha T^\text{r}}{\log (1 + \gamma^\text{s}_{t})}  N^{*}_{t},\,\, \forall t \in \mathcal{T},
    \end{equation}
where~$N^{*}_{t}$ denotes the minimum value of $N_{t}$, given by:
    \begin{equation}\label{}
      N^{*}_{t} =  \arg \min_{N_{t}} P( N_{t} \ge \sum_{f \in \mathcal{F}_{t}}{a_{t,f}} | \hat{\mathbf{a}}_{t}  ) \ge \varepsilon, 
    \end{equation}
where~$\hat{\mathbf{a}}_{t} = [\hat{a}_{t,f}]_{\forall f \in \mathcal{F}_{t}}$. To determine~$N^{*}_{t}$, we need to obtain the conditional probability~$P( N_{t} \ge \sum_{f \in \mathcal{F}_{t}}{a_{t,f}} | \hat{\mathbf{a}}_{t}  )$. Without loss of generality, we assume that, given~$\hat{a}_{t,f}, \forall f \in \mathcal{F}_{t}$, random variables~$a_{t,f}, \forall f \in \mathcal{F}_{t}$ are independent and identically distributed (i.i.d.), and~$P(a_{t,f} | \hat{\mathbf{a}}_{t}) = P(a_{t,f} | \hat{a}_{t,f}), \forall f \in \mathcal{F}_{t}$. Define the following three parameters: model accuracy performance~$p_{t} = P(\hat{a}_{t,f}=1 | a_{t,f}=1)$, ~$q_{t} = P(\hat{a}_{t,f}=0 | a_{t,f}=0)$, and key frame ratio~$\lambda_{t} = P(a_{t,f}=1)$.

    \begin{theorem}\label{theorem1}
         The probability~$P(N_{t} \ge \sum_{f \in \mathcal{F}_{t}}{a_{t,f}} | \hat{\mathbf{a}}_{t} )$ given prediction results from the M-UDT, can be derived in~\eqref{eq30}, which is non-decreasing, where $\hat{A} = \sum_{f \in \mathcal{F}_{t}}{\hat{a}_{t,f}}$, $F = |\mathcal{F}_{t}|$,
            \begin{equation}\label{}
               p_{t}^\text{TPR} = \frac{p_{t}\lambda_{t}}{p_{t}\lambda_{t} +(1-q_{t})(1-\lambda_{t})},
            \end{equation}
        and
            \begin{equation}\label{}
               p_{t}^\text{TNR} = \frac{q_{t}(1-\lambda_{t})}{q_{t}(1-\lambda_{t}) + (1-p_{t})\lambda_{t}}.
            \end{equation}
    \end{theorem}
    \begin{proof}   
        Omitted due to the limit of space.
    \end{proof}
Theorem~1 allows us to derive a closed-form solution of~$N^{*}_{t}$ given parameters~$p_{t}$,~$q_{t}$, and~$\lambda_{t}$. The three parameters, representing the management-oriented data stored in the MUT, can be \bl{updated per time slot according to} user-oriented data, i.e.,~$\hat{\mathbf{a}}_{t}$ and~$\mathbf{a}_{t}$ following a moving-average rule.

    \begin{figure}[t]
        \centering
        \includegraphics[width=0.44\textwidth]{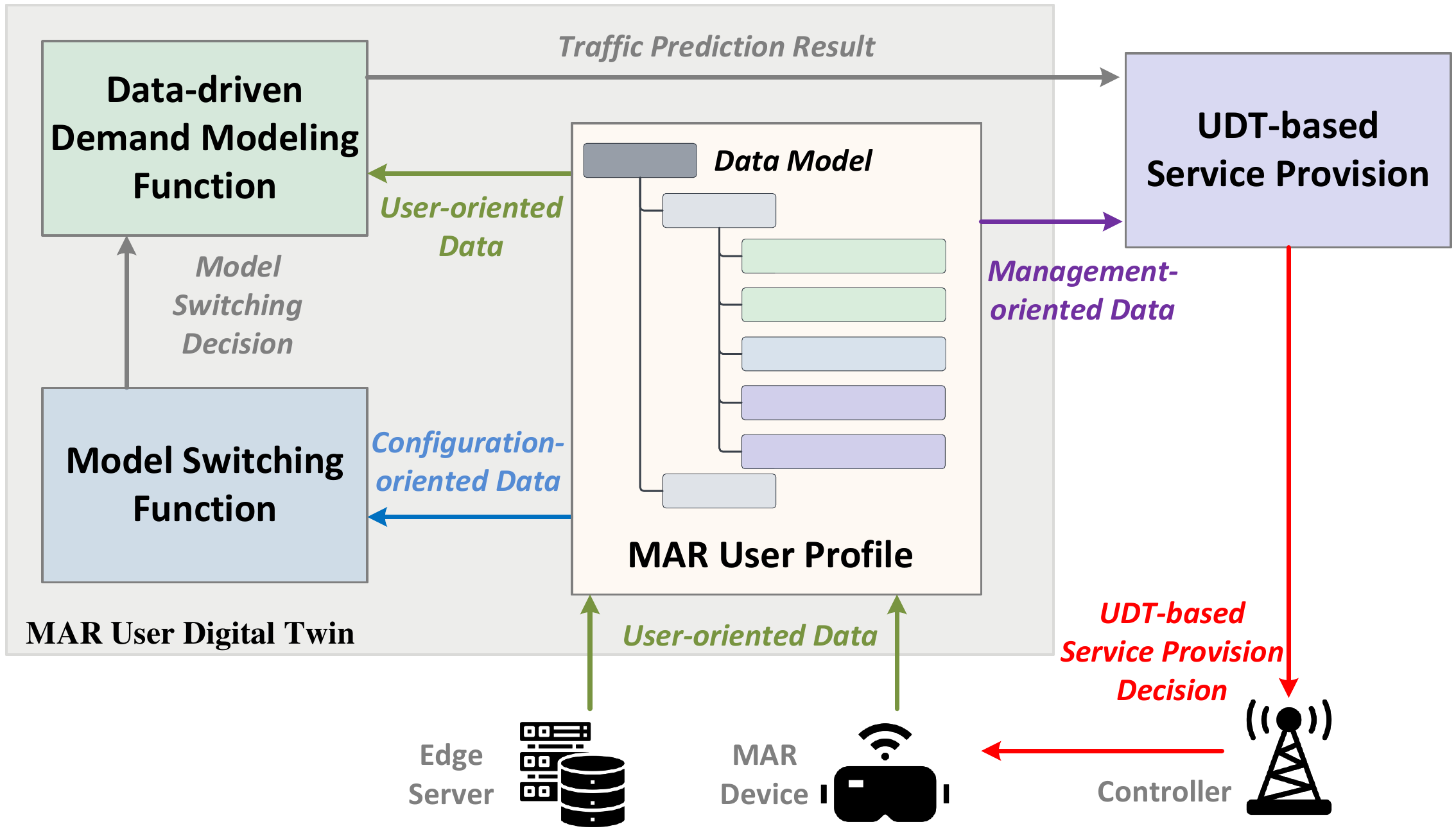}
        \caption{The workflow of the developed DT-based approach.}\label{udt}
    \end{figure} 

        \begin{figure*}[t] 
        \begin{equation}\label{eq30}
            \begin{split}
                & g(N_{t}; p_{t}, q_{t}, \lambda_{t}) = \sum_{k=0}^{N_{t}}{ \sum_{j = \max(0, k-(F-\hat{A}))}^{\min(\hat{A},k)}{\binom{\hat{A}}{j} (p_{t}^\text{TPR})^{j}(1-p_{t}^\text{TPR})^{\hat{A}-j} \binom{F-\hat{A}}{k-j} (1-p_{t}^\text{TNR})^{k-j}(p_{t}^\text{TNR})^{F-\hat{A}-k+j} }},
            \end{split}
        \end{equation}
        \rule[-10pt]{18.15cm}{0.05em}  
    \end{figure*}

We show the workflow of our \bl{M-UDT-based service provision approach} in Fig.~\ref{udt}. The MUP comprises the data model with structured user data essential for service provision. The designed DMF and MSF enable the data update in the MUP, thereby enabling the user-centric service provision. 

\section{Performance Evaluation}

\subsection{Simulation Settings}

In our simulation, we use 218 camera frame sequences, corresponding to different user movement in various environments, from the InteriorNet dataset~\cite{li2018interiornet} and conduct device pose tracking for the MAR device using the open-source ORB-SLAM3 platform~\cite{campos2021orb}. We use a resource block (RB) as the base unit for spectrum resource, each of which is 180\,kHz wide (12 subcarriers) in bandwidth and 0.5\,ms long in time. Other important parameter settings are listed in Table~I.

We adopt the following prevalent data traffic modeling approaches as benchmark:
    \begin{itemize}
        \item \emph{Poisson regression}: The number of key frames for uploading in each time slot is assumed to follow a Poisson distribution. The parameter of the Poisson distribution \bl{is} estimated based on historical information;

        \item \emph{LSTM neural network}: Following the simplified modeling in the DMF, an LSTM neural network is pre-trained and employed to predict the \bl{number} of key frames that need to be uploaded in each time slot.
    \end{itemize}

    \begin{table}[t]
        \footnotesize 
        \centering
        \captionsetup{justification=centering,singlelinecheck=false}
        \caption{Simulation Parameters}\label{table1}
        \begin{tabular}{c|c|c|c}
            \hline\hline
             Parameter & Value & Parameter & Value\\
             \hline\hline
             $F$ & 10\,frames & $T^\text{r}$ & 0.02\,second \\
             \hline
             $\alpha$ & 5\,Mbits & $\gamma_{t}^\text{s}$ & 15\,dB\\
             \hline
             $\delta$ & 4 & $M$ & 3\\
             \hline
        \end{tabular}
    \end{table}

\subsection{Performance of the M-UDT-based Approach}

    \begin{figure}[t]
        \centering
        \includegraphics[width=0.35\textwidth]{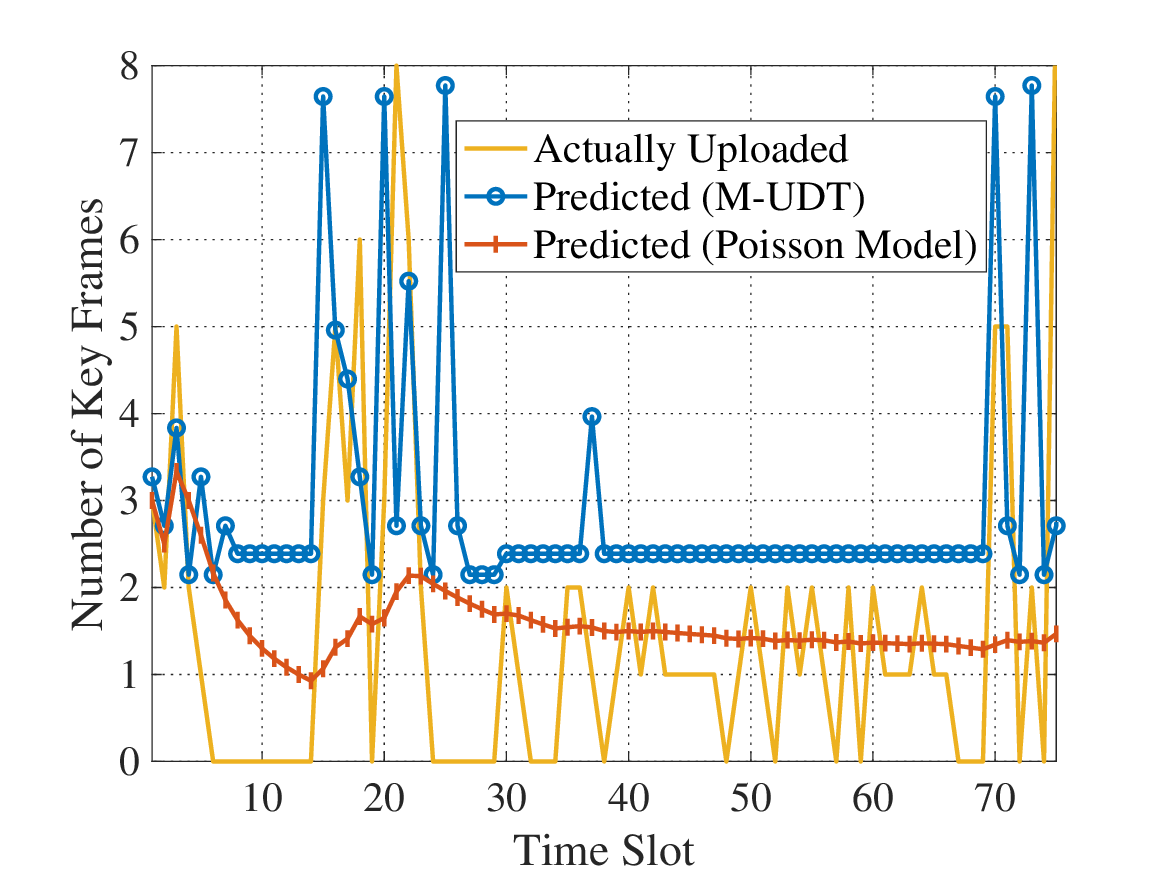}
        \caption{Data traffic modeling performance comparison.}\label{fig1}
    \end{figure}

In Fig.~\ref{fig1}, we compare the traffic modeling performance of M-UDT with that of Poisson regression, \bl{labeled as ``Predicted (Poisson Model)'',} over one camera frame sequence. We can observe that the predicted values by the M-UDT more closely match the actual non-stationary uplink data traffic, particularly during \bl{bursts} in uplink data traffic that may result from device tracking loss or changes in the physical environment. This is because the M-UDT can switch between detailed and simplified data-driven modeling according to variations in the number of uploaded key frames, thereby capturing \bl{the implicit impact of the MAR operational mechanism on data traffic load} while reducing input data redundancy in the detailed modeling.

In Fig.~\ref{fig2}, we compare the service provision performance of the M-UDT-based approach with that of the LSTM-based approach (labeled as ``LSTM'') in terms of spectrum resource utilization and delay satisfaction. Given different \bl{tolerable transmission duration} for uploading the selected key frames, i.e.,~${T}^\text{r}$, we plot the amount of over-provisioned spectrum resource (in RBs) in the two approaches. From the figure, we can observe that, due to the high accuracy of the \bl{M-UDT} in data traffic modeling, our M-UDT-based approach not only reduces \bl{the amount of over-provisioned} spectrum resource but also ensures the timeliness of key frame uploading for the MAR device, leading to advanced user-centric service provision.

    \begin{figure}[t]
        \centering
        \includegraphics[width=0.35\textwidth]{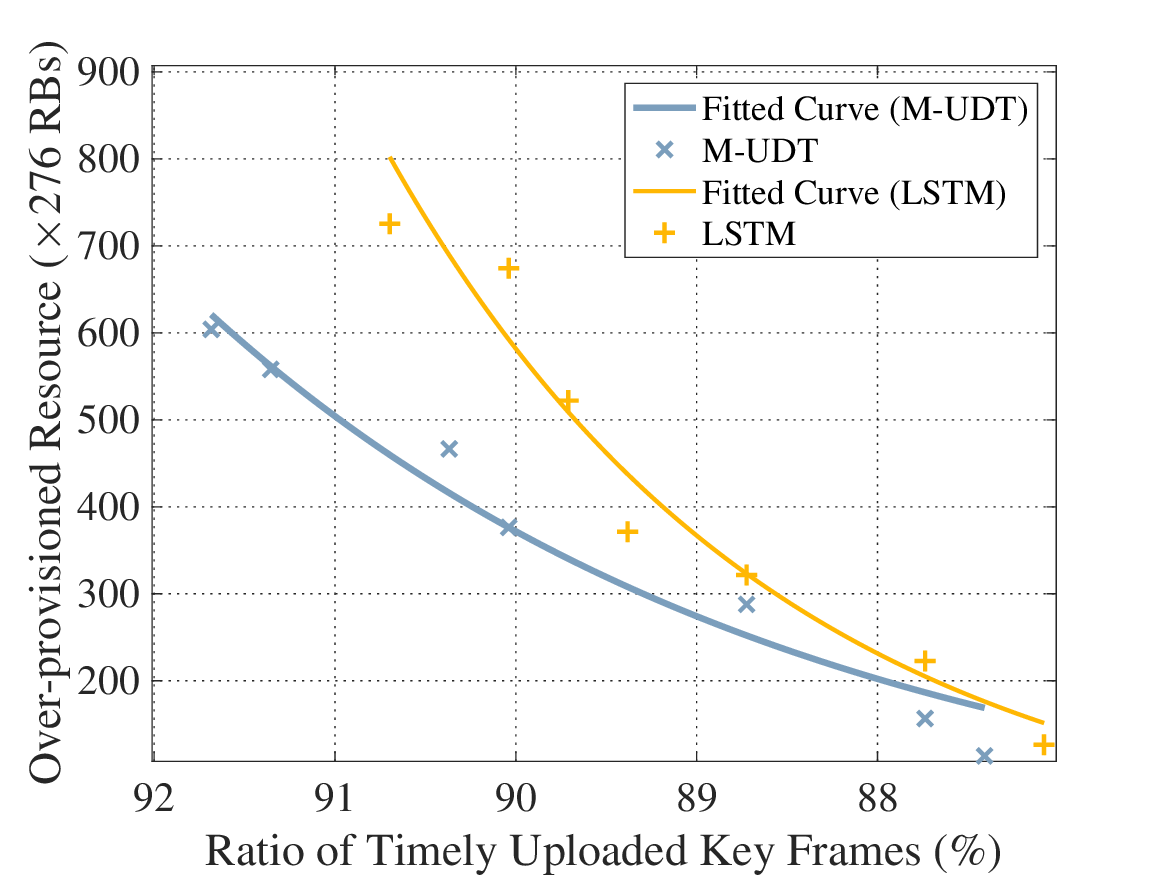}
        \caption{Service provision performance of the M-UDT-based approach.}\label{fig2}
    \end{figure}

\section{Conclusion and Future Work}

In this paper, we have developed a data-driven service provision approach based on the M-UDT to support customized user experiences in edge-assisted MAR. In the M-UDT, the established hierarchical data model organizes the factors affecting user-specific data traffic, and the designed UDT functions enable the switching between two data-driven traffic models to adapt to non-stationary data traffic. Simulation results have demonstrated the effectiveness of the developed M-UDT-based data-driven approach in reducing spectrum resource consumption while satisfying the delay requirement of camera frame uploading due to high modeling accuracy. Our approach provides a scalable and flexible paradigm to characterize the intricate impacts of MAR operational mechanisms on user-specific resource demands, which facilitates the shift to user-centric service provision in the 6G era. In the future, we plan to incorporate service provision for multiple MAR devices with diverse camera frame uploading mechanisms.

\bibliography{ref_AR2}

\bibliographystyle{IEEEtran}

\end{document}